\documentclass[letterpaper, 10 pt, conference]{ieeeconf}

\IEEEoverridecommandlockouts                              

\usepackage[top=0.833in, left=0.667in, right=0.667in, bottom=0.597in, includefoot]{geometry}

\usepackage{cite}

\usepackage{amsthm}

\newtheorem{lemma}{Lemma}

\newtheorem{remark}{Remark}
\newtheorem{theorem}{Theorem}

\newtheorem{definition}{Definition}

\newtheorem{problem}{Problem}

\usepackage{mathrsfs} 
\usepackage{eufrak}
\usepackage{url}
\usepackage{latexsym}

\usepackage{qtree}
\usepackage{subcaption}
\usepackage{listings}
\usepackage{graphicx}
\usepackage{amsmath}
\usepackage{amsfonts}
\usepackage{amssymb}
\usepackage{multirow}
\usepackage[table]{xcolor}
\usepackage{tikz}
\usepackage{pgfplots}
\usetikzlibrary{automata,arrows,positioning,chains,fit,shapes,calc} 
\usepgfplotslibrary{patchplots}
\tikzset{
    >=stealth',
    punkt/.style={
           rectangle,
           rounded corners,
           draw=black, very thick,
           text width=12em,
           minimum height=2em,
           text centered},
    pil/.style={
           ->,
           thick,
           shorten <=2pt,
           shorten >=2pt,}
}

\usepackage{algorithm}
\usepackage[noend]{algorithmic}   
\usepackage{listings}
\usepackage[symbols,toc]{glossaries-extra}

\makenoidxglossaries

\setabbreviationstyle[acronym]{long-short}
\preto\chapter{\glsresetall}

\newacronym{cegis}{CEGIS}{counterexample-guided inductive synthesis}
\newacronym{csp}{CSP}{constraint satisfiability problem}
\newacronym{smt}{SMT}{satisfiability modulo theories}
\newacronym{lp}{LP}{linear programming}
\newacronym{milp}{MILP}{mixed-integer linear programming}
\newacronym{ips}{IPS}{intelligent physical system}
\newacronym{ltl}{LTL}{linear temporal logic}
\newacronym{rtl}{RTL}{temporal logic over reals}
\newacronym{prtl}{PRTL}{probabilistic temporal logic over reals}
\newacronym{stl}{STL}{signal temporal logic}
\newacronym{mpc}{MPC}{model predictive control}
\newacronym{itmp}{ITMP}{integrated task and motion planning}
\newacronym{ai}{AI}{artificial intelligence}
\newacronym{ff}{FF}{fast forward}
\newacronym{idtmp}{IDTMP}{iteratively deepened task and motion planning}
\newacronym{cosmop}{CoSMoP}{composition of safe motion primitives}
\newacronym{mld}{MLD}{mixed logical dynamic}
\newacronym{pomdp}{POMDP}{partially observable Markov decision process}
\newacronym{prstl}{PrSTL}{probabilistic signal temporal logic}
\newacronym{socp}{SOCP}{second-order cone programming}
\newacronym{rhc}{RHC}{receding horizon control}
\newacronym{kf}{KF}{Kalman filter}
\newacronym{ukf}{UKF}{unscented Kalman filter}
\newacronym{ekf}{EKF}{extended Kalman filter}
\newacronym{smc}{SMC}{sequencial Monte-Carlo}
\newacronym{lqr}{LQR}{linear quadratic regulator}
\newacronym{zoh}{ZOH}{zero order hold}
\newacronym{ir}{IR}{infrared}
\newacronym{cps}{CPS}{cyber-physical systems}
\newacronym{dof}{DOF}{degrees of freedom}
\newacronym{rrt}{RRT}{Rapidly-exploring Random Tree}
\newacronym{ltlopt}{LTLOpt}{optimal control with linear temporal logic specifications}
\newacronym{ros}{ROS}{Robot Operating System}
\newacronym{bsc}{BSC}{Bounded Satisfiability Checking}
\newacronym{bmc}{BMC}{Bounded Model Checking}
\newacronym{ompl}{OMPL}{Open Motion Planning Library}
 
\glsxtrnewsymbol 
[description={logical tautology}]
{true}
{\ensuremath{\top}}

\glsxtrnewsymbol 
[description={logical contradiction}]
{false}
{\ensuremath{\perp}}

\glsxtrnewsymbol 
[description={logical conjunction}]
{and}
{\ensuremath{\wedge}}

\glsxtrnewsymbol 
[description={logical disjunction}]
{or}
{\ensuremath{\vee}}

\glsxtrnewsymbol 
[description={logical negation}]
{neg}
{\ensuremath{\neg}}

\glsxtrnewsymbol 
[description={logical conditional}]
{implies}
{\ensuremath{\rightarrow}}

\glsxtrnewsymbol 
[description={logical biconditional}]
{iff}
{\ensuremath{\leftrightarrow}}

\glsxtrnewsymbol
[description={the set containing no elements}]
{emptyset}
{\ensuremath{\emptyset}}

\glsxtrnewsymbol
[description={a propositional variable}]
{propvar}
{\ensuremath{p}}

\glsxtrnewsymbol 
[description={a predicate}]
{pred}
{\ensuremath{\pi}}

\glsxtrnewsymbol 
[description={a set of predicates}]
{predset}
{\ensuremath{\Pi}}

\glsxtrnewsymbol 
[description={assertion}]
{assert}
{\ensuremath{\leftarrow}}

\glsxtrnewsymbol 
[description={the indentity matrix}]
{eye}
{\ensuremath{I}}

\glsxtrnewsymbol 
[description={a matrix of zeros}]
{zeros}
{\ensuremath{\boldsymbol{0}}}

\glsxtrnewsymbol 
[description={a matrix of ones}]
{ones}
{\ensuremath{\boldsymbol{1}}}
\DeclareMathAlphabet\mathbfcal{OMS}{cmsy}{b}{n}

\glsxtrnewsymbol 
[description={a discrete transition map}]
{jump}
{\ensuremath{\hookrightarrow}}

\glsxtrnewsymbol 
[description={the vector of continuous state variables}]
{statec}
{\ensuremath{\boldsymbol{x}}}

\glsxtrnewsymbol 
[description={the domain of continuous state variables}]
{statecdom}
{\ensuremath{\mathcal{X}}}

\glsxtrnewsymbol 
[description={the vector of initial continuous state}]
{statecini}
{\ensuremath{\bar{\boldsymbol{x}}}}

\glsxtrnewsymbol 
[description={length of the vector of continuous state variables}]
{statecnb}
{\ensuremath{n}}

\glsxtrnewsymbol 
[description={a vector of continuous input variables}]
{inputc}
{\ensuremath{\boldsymbol{u}}}

\glsxtrnewsymbol 
[description={the domain of continuous input variables}]
{inputcdom}
{\ensuremath{\mathcal{U}}}

\glsxtrnewsymbol 
[description={length of the vector of continuous input variables}]
{inputcnb}
{\ensuremath{m}}

\glsxtrnewsymbol 
[description={the sampling period}]
{ts}
{\ensuremath{T_s}}

\glsxtrnewsymbol 
[description={a sequence of states and input variables evaluations}]
{run}
{\ensuremath{\boldsymbol{\xi}}}

\glsxtrnewsymbol 
[description={a sequence of states and input variables evaluations for a simulation model}]
{run_coarse}
{\ensuremath{\hat{\boldsymbol{\xi}}}}

\glsxtrnewsymbol 
[description={the continuous flow map}]
{f}
{\ensuremath{f}}

\glsxtrnewsymbol 
[description={the discrete jump map}]
{g}
{\ensuremath{g}}

\glsxtrnewsymbol 
[description={a continuous output map}]
{h}
{\ensuremath{h}}

\glsxtrnewsymbol 
[description={the discrete state}]
{stated}
{\ensuremath{q}}

\glsxtrnewsymbol 
[description={the initial discrete state}]
{statedini}
{\ensuremath{\bar{q}}}

\glsxtrnewsymbol 
[description={the finite set of discrete states}]
{statedset}
{\ensuremath{Q}}

\glsxtrnewsymbol 
[description={the discrete control input}]
{inputd}
{\ensuremath{e}}

\glsxtrnewsymbol 
[description={the finite set of discrete control inputs}]
{inputdset}
{\ensuremath{E}}

\glsxtrnewsymbol 
[description={the vector of output (observation) continuous variables}]
{outputc}
{\ensuremath{\boldsymbol{y}}}

\glsxtrnewsymbol 
[description={the length of the vector of output (observation) continuous variables}]
{outputcnb}
{\ensuremath{p}}

\glsxtrnewsymbol 
[description={the domain of continuous input variables}]
{outputcdom}
{\ensuremath{\mathcal{Y}}}

\glsxtrnewsymbol 
[description={the standard deviation matrix of a stochastic disturbance}]
{disturbance-std}
{\ensuremath{R}}

\glsxtrnewsymbol 
[description={the standard deviation matrix of a stochastic noise}]
{noise-std}
{\ensuremath{W}}

\glsxtrnewsymbol 
[description={the standard Gaussian (Normal) distribution}]
{stdnormal}
{\ensuremath{\boldsymbol{\upsilon}}}

\glsxtrnewsymbol 
[description={the Gaussian (Normal) distribution}]
{normal}
{\ensuremath{\mathcal{N}}}

\glsxtrnewsymbol 
[description={the vector of continuous state random variables}]
{prstatec}
{\ensuremath{\boldsymbol{X}}}

\glsxtrnewsymbol 
[description={the discrete state random variable}]
{prstated}
{\ensuremath{\boldsymbol{Q}}}

\glsxtrnewsymbol 
[description={the vector of continuous output random variables}]
{proutputc}
{\ensuremath{\boldsymbol{Y}}}

\glsxtrnewsymbol 
[description={the vector of parameters which defines the belief state}]
{belief}
{\ensuremath{\boldsymbol{b}}}

\glsxtrnewsymbol 
[description={the standard deviation variable}]
{std}
{\ensuremath{\sigma}}

\glsxtrnewsymbol 
[description={the vector of standard deviation variables}]
{stdvec}
{\ensuremath{\boldsymbol{\sigma}}}

\glsxtrnewsymbol 
[description={the mean variable}]
{mean}
{\ensuremath{m}}

\glsxtrnewsymbol 
[description={the vector of mean variables}]
{meanvec}
{\ensuremath{\boldsymbol{m}}}

\glsxtrnewsymbol 
[description={the matrix of standard deviation variables}]
{stdmat}
{\ensuremath{S}}

\glsxtrnewsymbol 
[description={the matrix of covariance variables}]
{covmat}
{\ensuremath{\Sigma}}

\glsxtrnewsymbol 
[description={the vector of the initial mean values}]
{meanvecini}
{\ensuremath{{\boldsymbol{m}_0}}}

\glsxtrnewsymbol 
[description={the matrix of the initial covariance values}]
{covmatini}
{\ensuremath{\Sigma_0}}

\glsxtrnewsymbol 
[description={the vector of parameters which defines the initial belief state}]
{beliefini}
{\ensuremath{\boldsymbol{b}_0}}

\glsxtrnewsymbol 
[description={the inovation covariance}]
{inovar}
{\ensuremath{\Gamma}}

\glsxtrnewsymbol 
[description={the cumulative distribution function of the standard Gaussian (Normal) distribution}]
{cdf}
{\ensuremath{\Phi}}

\glsxtrnewsymbol 
[description={the quantile function of the standard Gaussian (Normal) distribution}]
{icdf}
{\ensuremath{\Phi^{-1}}}

\glsxtrnewsymbol 
[description={the quantile function of an probabilistic atomic proposition}]
{norminv}
{\ensuremath{\lambda}}

\glsxtrnewsymbol 
[description={a discrete observation}]
{outputd}
{\ensuremath{o}}

\glsxtrnewsymbol 
[description={a set of discrete observation}]
{outputdset}
{\ensuremath{\Omega}}

\glsxtrnewsymbol 
[description={a set of conditional discrete observation probabilities}]
{outputdcprob}
{\ensuremath{O}}

\glsxtrnewsymbol 
[description={a labeled-\gls*{pomdp} with hybrid dynamics}]
{hpomdp}
{\ensuremath{\mathfrak{H}}}

\glsxtrnewsymbol 
[description={a conditional transition probabilities between discrete states}]
{statedcprob}
{\ensuremath{T}}

\glsxtrnewsymbol 
[description={a set of conditional transition probabilities between continuous states}]
{stateccprob}
{\ensuremath{G}}

\glsxtrnewsymbol 
[description={a non-negative scalar representing a Gaussian mixture weigh}]
{weight}
{\ensuremath{w}}

\glsxtrnewsymbol 
[description={the number of Gaussian mixture components}]
{gmmnb}
{\ensuremath{N_g}}

\glsxtrnewsymbol 
[description={a non-negative scalar representing an initial Gaussian mixture weigh}]
{weightini}
{\ensuremath{\bar{w}}}

\glsxtrnewsymbol 
[description={the initial number of Gaussian mixture components}]
{gmmnbini}
{\ensuremath{\bar{n}}}

\glsxtrnewsymbol 
[description={the probability of a discrete random variable}]
{dprob}
{\ensuremath{p}}

\glsxtrnewsymbol 
[description={the initial probability of a discrete random variable}]
{dprobini}
{\ensuremath{\bar{p}}}

\glsxtrnewsymbol 
[description={an \gls*{stl} formula}]
{stlformula}
{\ensuremath{\varphi}}

\glsxtrnewsymbol 
[description={a linear affine function}]
{stlfunc}
{\ensuremath{\mu}}

\glsxtrnewsymbol 
[description={a binary logic operator which states that an interpretation satisfies a formula}]
{sat}
{\ensuremath{\vDash}}

\glsxtrnewsymbol 
[description={a binary logic operator which states that an interpretation does not satisfies a formula}]
{unsat}
{\ensuremath{\not\vDash}}

\glsxtrnewsymbol 
[description={the temporal logic operator nex}]
{next}
{\ensuremath{\bigcirc}}

\glsxtrnewsymbol 
[description={the temporal logic operator always}]
{always}
{\ensuremath{\square}}

\glsxtrnewsymbol 
[description={the temporal logic operator eventually}]
{eventually}
{\ensuremath{\Diamond}}

\glsxtrnewsymbol 
[description={the temporal logic operator until}]
{until}
{\ensuremath{\boldsymbol{U}}}

\glsxtrnewsymbol 
[description={the temporal logic operator release}]
{release}
{\ensuremath{\boldsymbol{R}}}

\glsxtrnewsymbol 
[description={a real-valued function which associates a scalar value with the quality of satisfaction of an \gls*{stl} formula}]
{robfunc}
{\ensuremath{\rho}}

\glsxtrnewsymbol 
[description={a sequence of belief states and input variables evaluations}]
{runpr}
{\ensuremath{\boldsymbol{\beta}}}

\glsxtrnewsymbol 
[description={a sequence of initial belief states and input variables evaluations}]
{runprini}
{\ensuremath{\bar{\boldsymbol{\beta}}}}

\glsxtrnewsymbol 
[description={a sequence of states and input variables evaluations for a simulation model}]
{runpr_coarse}
{\ensuremath{\hat{\boldsymbol{\beta}}}}

\glsxtrnewsymbol 
[description={the probability tolerance}]
{tolpr}
{\ensuremath{\epsilon}}

\glsxtrnewsymbol 
[description={probabilistic logical negation}]
{negpr}
{\ensuremath{{\neg}}}

\glsxtrnewsymbol 
[description={a finite set of variables}]
{cspvarset}
{\ensuremath{V}}

\glsxtrnewsymbol 
[description={a variable}]
{cspvar}
{\ensuremath{v}}

\glsxtrnewsymbol 
[description={a finite set of domains of values}]
{cspdomset}
{\ensuremath{D}}

\glsxtrnewsymbol 
[description={a domainsof values}]
{cspdom}
{\ensuremath{\mathcal{D}}}

\glsxtrnewsymbol 
[description={a finite set of constraints}]
{cspconsset}
{\ensuremath{C}}

\glsxtrnewsymbol 
[description={a constraint}]
{cspcons}
{\ensuremath{c}}

\glsxtrnewsymbol 
[description={a real-valued function which associates a scalar value with the quality of satisfaction of an \gls*{stl} formula}]
{cspprob}
{\ensuremath{\langle \gls*{cspvarset}, \gls*{cspdomset}, \gls*{cspconsset} \rangle}}

\glsxtrnewsymbol 
[description={tolerance for dynamical feasibility}]
{tolfeas}
{\ensuremath{\delta}}

\glsxtrnewsymbol 
[description={a time indexed sequence of feedback controllers}]
{fctrlerseq}
{\ensuremath{\boldsymbol{F}}}

\glsxtrnewsymbol 
[description={a feedback controller}]
{fctrler}
{\ensuremath{F}}

\glsxtrnewsymbol 
[description={a time indexed sequence of convex constraints}]
{polyseq}
{\ensuremath{\boldsymbol{\mathbfcal{P}}}}

\glsxtrnewsymbol 
[description={a time indexed sequence of time constraints}]
{tconseq}
{\ensuremath{\boldsymbol{\mathbfcal{T}}}}

\glsxtrnewsymbol 
[description={a time constraints}]
{tcon}
{\ensuremath{\mathbfcal{T}}}

\glsxtrnewsymbol 
[description={a convex constraint}]
{poly}
{\ensuremath{\mathcal{P}}}

\glsxtrnewsymbol 
[description={an abstraction function}]
{abs}
{\ensuremath{\alpha}}

\glsxtrnewsymbol 
[description={a transition system}]
{tsys}
{\ensuremath{TS}}

\glsxtrnewsymbol 
[description={a transition of a transition system}]
{transition}
{\ensuremath{\rightarrow}}

\glsxtrnewsymbol 
[description={a label funciton of a transition system}]
{labelfunc}
{\ensuremath{\mathcal{L}}}

\glsxtrnewsymbol 
[description={a time length}]
{tlen}
{\ensuremath{\ell}}

\glsxtrnewsymbol 
[description={a vector of time lengths}]
{tlenvec}
{\ensuremath{\boldsymbol{\ell}}}

\glsxtrnewsymbol 
[description={the prefix index}]
{preidx}
{\ensuremath{N}}

\glsxtrnewsymbol 
[description={the discrete plan robust measure function}]
{dplanrobfunc}
{\ensuremath{\nu}}

\glsxtrnewsymbol 
[description={the conservative maximum time length required to decide the satisfiability of a \gls*{stl} formula}]
{kmax}
{\ensuremath{K_{\max}^{\gls*{stlformula}}}}

\glsxtrnewsymbol 
[description={a slack variable}]
{slack}
{\ensuremath{s}}

\glsxtrnewsymbol 
[description={a constant vector}]
{cvec}
{\ensuremath{\boldsymbol{h}}}

\glsxtrnewsymbol 
[description={a constant scalar}]
{cscalar}
{\ensuremath{a}}

\glsxtrnewsymbol 
[description={robustness measurement tolerance}]
{tolrob}
{\ensuremath{\theta}}

\newacronym{idstl}{idSTL}{iterative deepening Signal Temporal Logic}
\newacronym{idprstl}{idPrSTL}{iterative deepening probabilistic Signal Temporal Logic}
\newacronym{idprtl}{idPRTL}{iterative deepening probabilistic temporal logic over reals}
\newacronym{hstl}{HSTL}{Hybrid Signal Temporal Logic}
\newacronym{hprstl}{HPrSTL}{Hybrid Probabilistic Signal Temporal Logic}
\newacronym{mlo}{MLO}{maximum likelihood observation}
\newacronym{lqg}{LQG}{Linear Quadratic Gaussian regulator}

\usepackage{xifthen}
\makeatletter
\newcommand{\pushright}[1]{\ifmeasuring@#1\else\omit\hfill$\displaystyle#1$\fi\ignorespaces}
\newcommand{\pushleft}[1]{\ifmeasuring@#1\else\omit$\displaystyle#1$\hfill\fi\ignorespaces}
\makeatother
\DeclareMathOperator*{\argmin}{arg\,min}

\newcommand{\eye}[4]
{   \draw[rotate around={#4:(#2,#3)}] (#2,#3) -- ++(-.5*55:#1) (#2,#3) -- ++(.5*55:#1);
    \draw (#2,#3) ++(#4+55:.75*#1) arc (#4+55:#4-55:.75*#1);
    \draw[fill=gray] (#2,#3) ++(#4+55/3:.75*#1) arc (#4+180-55:#4+180+55:.28*#1);
    \draw[fill=black] (#2,#3) ++(#4+55/3:.75*#1) arc (#4+55/3:#4-55/3:.75*#1);
}

\definecolor{shadecolor}{RGB}{150,150,150}
\newcommand{\hl}{{}}

\newcommand{\idstlpred}[1][]{\ensuremath{\gls*{pred}^{\ifthenelse{\isempty{#1}}{\gls*{stlfunc}}{#1}}}}
\newcommand{\idstlpredint}[1][]{\ensuremath{\gls*{pred}^{\ifthenelse{\isempty{#1}}{\gls*{inteq}}{#1}}}}

\newcommand{\idstlpredpr}[1][]{\ensuremath{\gls*{pred}\ifthenelse{\isempty{#1}}{_{\gls*{tolpr}}^{\gls*{stlfunc}}}{#1}}}

\makeatletter
\let\NAT@parse\undefined
\makeatother
\usepackage{hyperref}  

\title{\LARGE \bf Active Perception and Control from Temporal Logic Specifications}
\author{
Rafael~Rodrigues~da~Silva, Vince~Kurtz, and~Hai~Lin
\thanks{The partial support of the National Science Foundation (Grant No.
ECCS-1253488, IIS-1724070, CNS-1830335) and of the Army Research
Laboratory (Grant No. W911NF- 17-1-0072) is gratefully acknowledged.}
\thanks{The authors are with the Department of Electrical Engineering, University of Notre Dame, Notre Dame, IN, 46556 USA. \url{{rrodri17,vkurtz,hlin1}@nd.edu}.}
}

\begin{document}

\maketitle
\thispagestyle{empty}
\pagestyle{empty}

\begin{abstract}
Next-generation autonomous systems must execute complex tasks in uncertain environments. Active perception, where an autonomous agent selects actions to increase knowledge about the environment, has gained traction in recent years for motion planning under uncertainty. One prominent approach is planning in the belief space. However, most belief-space planning starts with a known reward function, which can be difficult to specify for complex tasks. On the other hand, symbolic control methods automatically synthesize controllers to achieve logical specifications, but often do not deal well with uncertainty. In this work, we propose a framework for scalable task and motion planning in uncertain environments that combines the best of belief-space planning and symbolic control. Specifically, we provide a counterexample-guided-inductive-synthesis algorithm for probabilistic temporal logic over reals (PRTL) specifications in the belief space. Our method automatically generates actions that improve confidence in a belief when necessary, thus using active perception to satisfy PRTL specifications. 
\end{abstract}


%
\IEEEpeerreviewmaketitle

\section{Introduction}

Taskable and adaptive Intelligent Physical Systems (IPS) must not only take actions to satisfy high-level task specifications, but also change their behavior over time, learning from experience. These systems must formulate beliefs about the environment and explicitly act to improve confidence in these beliefs: this process is known as \textit{active perception} \cite{bajcsy1988active}. 

Active perception can be formalized as a partially observable Markov decision process (POMDP) \cite{krishnamurthy2016partially} since physical systems often follow the Markov assumption. POMDP planning involves a search for a policy that satisfies certain requirements. Since states are not directly observable, the policy is a mapping from history, a sequence of observations and actions, to new actions. History can be compactly represented as a probability distribution called the \textit{belief}.

Various methods have been proposed for active perception via belief space planning, most of which search for a policy that maximizes the expected value of a reward function \cite{bai2014integrated,valencia2013planning,agha2014firm,van2012efficient,indelman2016towards,platt2017efficient}.  
Specifying a reward function that guarantees completion of a complex task can be difficult, however. Using temporal logic specifications is often clearer and more intuitive. Furthermore, \textit{symbolic control} approaches, which synthesize controllers to achieve logical specifications, can provide strong formal guarantees. \hl{Early efforts in symbolic control focused on discrete transition system models. Probabilistic extensions to these models focus on uncertainty arising from state transitions, and mature software tools have been developed to this end \cite{ciesinski2004probabilistic,kwiatkowska2011prism}. In this work, we consider symbolic control of uncertain continuous systems. Specifically, we propose a provably correct framework for \textit{symbolic control in the belief space} based on Probabilistic Temporal Logic over Reals (PRTL) specifications. We focus on PRTL in particular because it is defined over real-valued signals and for its simple, clean notation.}

Existing synthesis methods for uncertain continuous systems are primarily based on mixed integer programming \cite{sadigh2015safe,zhong2017fast,jha2018safe}. These methods provide satisfying controllers for a convex fragment of \gls*{prstl}, another temporal logic for real-valued state variables. \hl{Other methods propose relaxations such as sampling-based optimization \cite{dey2016fast} and shrinking horizon model predictive control \cite{farahani2018shrinking} to achieve (possibly non-convex) specifications. However, these approaches focus on robustness to uncertainty rather than active perception. This means that such algorithms do not plan to gather more information, though doing so may be necessary to achieve a specification over the belief.}

We propose a PRTL controller synthesis algorithm for systems with perception and actuation uncertainty. \hl{By satisfying specifications defined over the belief space, our method incorporates active perception.} This means that the system not only satisfies expressive temporal specifications, but also synthesizes actions to gain information when necessary. \hl{We prove that our approach is sound and probabilistically complete, and demonstrate its effictiveness with a simulated example}.
 
The rest of this letter is organized as follows. First, we present the problem statement in Section \ref{sec:pr_preliminaries}. Next, we present the details of our controller synthesis algorithm in Section \ref{sec:hpomdp}. We show that our approach is sound and probabilistically complete in Section \ref{sec:correctness}. Section \ref{sec:mot_ex} provides an example of applying our methods to automated infrastructure inspection with a quadrotor, and Section \ref{sec:conclusion} concludes the paper.  

\section{Preliminaries}\label{sec:pr_preliminaries}

%

\subsection{System}\label{sec:prsystem}  

Consider the discrete-time linear control system 
\begin{equation}\label{eq:prsystem}
\begin{aligned}
	\gls*{statec}_{k+1} = & A\gls*{statec}_k + B\gls*{inputc}_k + \sqrt{\gls*{disturbance-std}} \gls*{stdnormal}  \\
	\gls*{outputc}_k = & C \gls*{statec}_k + \sqrt{\gls*{noise-std}_k} \gls*{stdnormal}, & \gls*{stdnormal} \sim \gls*{normal}(0,\gls*{eye}_{n}),
\end{aligned}  
\end{equation}
where $\gls*{statec} \in \mathbb{R}^{\gls*{statecnb}}$ are the state variables, $\gls*{inputc}  \in \gls*{inputcdom} \subset \mathbb{R}^{\gls*{inputcnb}}$ are the control inputs, $\gls*{outputc}  \in \mathbb{R}^{\gls*{outputcnb}}$ are the output variables, and $A \in \mathbb{R}^{\gls*{statecnb} \times \gls*{statecnb}}$, $B \in \mathbb{R}^{\gls*{statecnb}\times \gls*{inputcnb}}$, $C \in \mathbb{R}^{\gls*{outputcnb}\times \gls*{statecnb}}$ are constant matrices. We assume that $\gls*{inputcdom} = \{H_u \gls*{inputc} \geq \boldsymbol{c}_u \}$ is a full-dimensional polytope and that $(A,B)$ is stabilisable. 
This system is subject to uncorrelated Gaussian disturbances and noise with covariances $\gls*{disturbance-std}^\intercal \geq 0$ and $\gls*{noise-std}_k^\intercal \geq 0$, where $\gls*{noise-std}_k$ can be state dependent.  


\subsection{Belief State}\label{sec:belief}

Since only noisy observations $\gls*{outputc}_k$ are available in System (\ref{eq:prsystem}), we must estimate the state $\gls*{statec}_k$. To do so, the controller keeps track of a \textit{history} of observations and actions. History can be compactly represented as a random process $\gls*{belief}_k = P(\gls*{statec}_k)$ known as the \textit{belief state} \cite{krishnamurthy2016partially}. The belief state can be tracked using Bayesian filtering: 
\begin{equation*}
P(\gls*{statec}_{k+1}) = \eta P(\gls*{outputc}_{k+1}|\gls*{statec}_{k+1},\gls*{inputc}_k)\int_{\gls*{statec}} P(\gls*{statec}_{k+1}|\gls*{statec},\gls*{inputc}_k)\gls*{belief}_k, 
\end{equation*}
where $\eta$ is a normalization constant \cite{krishnamurthy2016partially}. 

In our case, we assume that the belief state is Gaussian with mean $\gls*{meanvec} _k \in \mathbb{R}^{\gls*{statecnb}}$ and covariance $\gls*{covmat}_k \in \mathbb{R}^{\gls*{statecnb} \times \gls*{statecnb}}$. Under this assumption, the belief dynamics can be derived by applying the Kalman Filter \cite{chui2017kalman}:
\begin{equation}\label{eq:beliefsys}
\gls*{meanvec}_{k+1} = 
\boldsymbol{K}_k\big(\gls*{outputc}_k - C(\gls*{f}_k)\big) + \gls*{f}_k, \quad \gls*{covmat}_{k+1} = \gls*{inovar}_k - \boldsymbol{K}_k C \gls*{inovar}_k, 
\end{equation}
where $\gls*{f}_k = A \gls*{meanvec} _k + B \gls*{inputc}_k$, $\boldsymbol{K}_k = \gls*{inovar}_k C^\intercal (C \gls*{inovar}_k C^\intercal + \gls*{noise-std}_k)^{-1}$ is the Kalman gain, and $\gls*{inovar}_k = A \gls*{covmat}_k A^\intercal  + \gls*{disturbance-std}$. 

Note that the belief update is a function of the measured observations $\gls*{outputc}_{k}$, which are unknown during planning. To address this issue, we follow \cite{platt2010belief} in planning according to the maximum likelihood observation (MLO) assumption. This gives rise to the simplified belief dynamics
\begin{equation}\label{eq:beliefsyssimple}
\gls*{meanvec}_{k+1} = A\gls*{meanvec}_k + B\gls*{inputc}_k, \quad \gls*{covmat}_{k+1} = \gls*{inovar}_k - \boldsymbol{K}_k C \gls*{inovar}_k.
\end{equation} 
 
Other approaches such as \cite{van2012efficient,indelman2014planning,indelman2016towards} avoid the \gls*{mlo} assumption by considering $\gls*{meanvec}_{k+1}$ as a random variable, i.e., $\gls*{meanvec}_{k+1} \sim \gls*{normal}(\gls*{f}_k,\boldsymbol{K}_k C \gls*{inovar}_k)$. \hl{While we focus on belief-space planning under the \gls*{mlo} assumption, we expect that extending our approach to more sophisticated belief dynamics will be relatively straightforward.}
   
\subsection{Probabilistic Temporal Logic over Reals}

We adopt an extension of temporal logic over reals (RTL) \cite{lin2014mission} to formally represent task requirements for System (\ref{eq:prsystem}). \hl{We first define a run of this system as a sequence $\rho = \gls*{belief}_0\gls*{inputc}_0\gls*{outputc}_0\gls*{belief}_1\dots$ with an initial belief state $\gls*{beliefini} \sim \gls*{normal}(\gls*{meanvecini},\gls*{covmatini})$. A sequence of belief states in a run $\rho$ is called a path: $\gls*{runpr} = \gls*{belief}_0\gls*{belief}_1\gls*{belief}_2\dots$}. This definition allows us to formally design specifications over $\gls*{runpr}$ using \gls*{prtl}.    

\gls*{prtl} formulas are defined recursively over a finite set of predicates $\gls*{predset}$ according to the following grammar: 
\begin{equation*}
    \gls*{stlformula} := \idstlpredpr | \gls*{negpr} \idstlpredpr | \gls*{stlformula}_1 \gls*{and} \gls*{stlformula}_2 | \gls*{stlformula}_1 \gls*{or} \gls*{stlformula}_2 | \gls*{stlformula}_1 \gls*{until} \gls*{stlformula}_2 | \gls*{stlformula}_1 \gls*{release} \gls*{stlformula}_2,  
\end{equation*}
where $\idstlpredpr \in \gls*{predset}$ is a probabilistic atomic predicate and $\gls*{stlformula}$, $\gls*{stlformula}_1$, $\gls*{stlformula}_2$  are \gls*{prtl} formulas. Each predicate is determined by a tolerance $\gls*{tolpr} \in [0,1]$ and the sign of the function  $\gls*{stlfunc}(\gls*{runpr}) = \gls*{cscalar} - \gls*{cvec}^\intercal \gls*{statec} $ (where $\gls*{cvec} \in \mathbb{R}^{\gls*{statecnb}}$ and $\gls*{cscalar} \in \mathbb{R}$). We denote the fact that a path $\gls*{runpr}$ satisfies a \gls*{prtl} formula $\gls*{stlformula}$ with $\gls*{runpr} \gls*{sat} \gls*{stlformula}$. \hl{Intuitively, this means that the trajectory of System (\ref{eq:prsystem}) fulfills the desired properties encoded in the specification $\gls*{stlformula}$.}

We write $\gls*{runpr} \gls*{sat}_{k} \gls*{stlformula}$ if the path $\gls*{belief}_k\gls*{belief}_{k+1}\dots$ satisfies $\gls*{stlformula}$.  Formally, the following semantics define the validity of a formula $\gls*{stlformula}$ with respect to the path $\gls*{runpr}$, where $\gls*{runpr} \gls*{sat} \gls*{stlformula}$ if and only if $\gls*{runpr} \gls*{sat}_0 \gls*{stlformula}$ and:   
\begin{itemize}
  \item $\gls*{runpr} \gls*{sat}_{k} \idstlpredpr$ if and only if $P(\gls*{stlfunc}(\gls*{statec}_k) \geq 0) > 1 - \gls*{tolpr}$,
  \item $\gls*{runpr} \gls*{sat}_{k} \gls*{negpr} \idstlpredpr$ if and only if $P(-\gls*{stlfunc}(\gls*{statec}_k) \geq 0) > 1 - \gls*{tolpr}$,
  \item $\gls*{runpr} \gls*{sat}_{k} \gls*{stlformula}_1 \gls*{and} \gls*{stlformula}_2$ if and only if $\gls*{runpr} \gls*{sat}_{k} \gls*{stlformula}_1$ and $\gls*{runpr} \gls*{sat}_{k} \gls*{stlformula}_2$,
  \item $\gls*{runpr} \gls*{sat}_{k} \gls*{stlformula}_1 \gls*{or} \gls*{stlformula}_2$ if and only if $\gls*{runpr} \gls*{sat}_{k} \gls*{stlformula}_1$ or $\gls*{runpr} \gls*{sat}_{k} \gls*{stlformula}_2$,
  \item $\gls*{runpr} \gls*{sat}_{k} \gls*{stlformula}_1 \gls*{until} \gls*{stlformula}_2$ if and only if $\exists {k^\prime} \geq k$ s.t. $\gls*{runpr} \gls*{sat}_{{k^\prime}}\gls*{stlformula}_2$, and $\forall k \leq {k^{\prime\prime}} \leq {k^\prime}\, \gls*{runpr} \gls*{sat}_{{k^{\prime\prime}}}\gls*{stlformula}_1$,
  \item $\gls*{runpr} \gls*{sat}_{k} \gls*{stlformula}_1 \gls*{release} \gls*{stlformula}_2$ if and only if $\exists {k^\prime} \geq k$ s.t. $\gls*{runpr} \gls*{sat}_{{k^\prime}}\gls*{stlformula}_1$, and $\gls*{runpr} \gls*{sat}_{{k^{\prime\prime}}}\gls*{stlformula}_2\, \forall k \leq {k^{\prime\prime}} \leq {k^\prime}$ or $\gls*{runpr} \gls*{sat}_{{k^\prime}}\gls*{stlformula}_2\, \forall t_{k^\prime} \geq k$, 
\end{itemize}  

We can derive other temporal operators such as \textit{eventually} $\gls*{eventually} \gls*{stlformula} = \gls*{true} \gls*{until} \gls*{stlformula}$ and \textit{always} $\gls*{always} \gls*{stlformula} = \gls*{false} \gls*{release} \gls*{stlformula}$. \hl{For example, the specification that ``$\gls*{statec}$ should be positive with at least $95\%$ probability" can be encoded as $\varphi = \gls*{always} \idstlpredpr[_{0.05}^{\gls*{statec}}]$, where $\gls*{stlfunc}(\gls*{runpr}) =\gls*{statec}$.}

\subsection{Problem Formulation}

The problem of task and motion planning with active perception for System (\ref{eq:prsystem}) can now be formulated in terms of a \gls*{prtl} specification and belief space dynamics. This problem is formally defined as follows:

\begin{problem}\label{prob:pr_1}
Given a stochastic system (\ref{eq:prsystem}), a \gls*{prtl} formula $\gls*{stlformula}$, and a prior belief state $\gls*{beliefini} \sim \gls*{normal}(\gls*{meanvecini},\gls*{covmatini})$, determine whether there exists a path $\gls*{runpr}$ of System (\ref{eq:prsystem}) that satisfies $\gls*{stlformula}$  and return the corresponding control inputs  $\gls*{inputc}_0\gls*{inputc}_1\gls*{inputc}_2\dots$.
\end{problem}

\section{Proposed Approach}\label{sec:hpomdp}

Our proposed approach, \gls*{idprtl}, is illustrated in Figure \ref{fig:diag2}. We reformulate the planning problem as an \textit{existential model checking} problem and use \textit{counterexample-guided synthesis} \cite{reynolds2015counterexample} to find a path that satisfies \gls*{prtl} specification $\gls*{stlformula}$. 

Specifically, two interacting layers, discrete and continuous, work together to overcome nonconvexities in the logical specification. At the discrete layer, existential \gls*{bmc} for an abstraction of System (\ref{eq:prsystem}) acts as a \textit{proposer}, generating a discrete path that satisfies the specification. Satisfying discrete plans are passed to the continuous layer, which acts as a \textit{teacher}. At the continuous layer, a sampling-based search is applied to check whether a discrete plan is feasible. If the feasibility test does not pass, a counterexample is provided to update the abstraction. This forces the discrete layer to propose a different discrete path. \hl{This process repeats until either a feasible plan is found or there is no satisfying discrete plan.}

\begin{figure}
\tikzstyle{block} = [draw, rounded corners=1mm, color=blue, text=black, line width=0.5mm, rectangle, minimum height=3em, minimum width=9.9em]
\tikzstyle{block2} = [draw, color=black, text=black, line width=0.5mm, rectangle, minimum height=3em, minimum width=9em]
\centering
\begin{tikzpicture}[auto, >=latex', scale=0.85, transform shape]	
\node[block] (dplan) {\begin{tabular}{c}Existential \\ Model Checking\end{tabular} };
\node[block,right=2.5cm of dplan] (fsearch) {\begin{tabular}{c}Feasibility \\ Search\end{tabular} }; 
\draw[->] ([yshift=1mm] dplan.east) -- node[pos=0.5,above] {Discrete path} ([yshift=1mm] fsearch.west); 
\draw[->] ([yshift=-1mm] fsearch.west) -- node[pos=0.5,below] {Counterexample} ([yshift=-1mm] dplan.east); 
\draw[->] (fsearch.south) -- node[pos=0.9,below] {Nominal run} ++(0,-0.5cm); 
\draw[->] (dplan.south) -- node[pos=0.9,below] {Unsatisfiable} ++(0,-0.5cm); 


\end{tikzpicture}
\caption{Bounded Existential Model Checking (BMC) proposes runs of an abstract system, while sampling-based feasibility search finds a corresponding run of the actual system. Feasibility search uses belief-space planning methods, incorporating active perception. } 
\label{fig:diag2}
\end{figure}
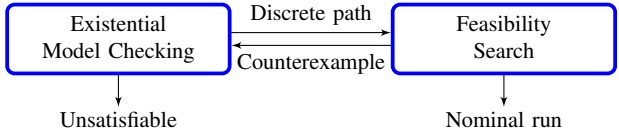

\subsection{Discrete Existential Model Checking}

We aim to find a path \gls*{runpr} of System (\ref{eq:prsystem}) that satisfies a \gls*{prtl} formula. To do so, we propose a \hl{finite} abstraction that captures the requirements of \gls*{stlformula}. \hl{The finite nature of this abstraction allows us to use counterexample-guided-synthesis to find a dynamically feasible discrete path}. We propose this abstraction as a deterministic Kripke structure: 

\begin{definition}\cite{biere1999symbolic}
A Kripke structure is a tuple $M = (S,I,\rightarrow,\mathcal{L})$, where $S$ is a finite set of states, $I \subseteq S$ are initial states, $\rightarrow \subseteq S \times S$ is a transition relation, and $\mathcal{L}:S \mapsto 2^{AP}$ is a labeling function which maps to atomic propositions $AP$.  
\end{definition}

To define a Kripke abstraction for System (\ref{eq:prsystem}), first consider the subspace $\gls*{poly}_i = \{ \gls*{belief} |  \bigwedge_{\idstlpredpr \in \gls*{predset}_i} \idstlpredpr\}$, which defines a region of the belief space where all atomic propositions $\gls*{predset}_i \subseteq \gls*{predset}$ hold. Considering all subsets of $\gls*{predset}$, we can construct the space $\{\gls*{poly}_1, \dots, \gls*{poly}_{2^{|\gls*{predset}|}} \}$. We now define the following abstraction:

\begin{definition}\label{def:abs_operator}
The abstraction of a belief state $\gls*{belief}$ is the set of all subspaces $\gls*{poly}_i$ that contain $\gls*{belief}$: $$\gls*{abs}(\gls*{belief}) = \{ \gls*{poly}_i \in \{\gls*{poly}_1, \dots, \gls*{poly}_{2^{|\gls*{predset}|}} \}\mid \gls*{belief} \in \gls*{poly}_i \}$$
\end{definition}
   
\hl{Each subspace $\gls*{poly}_i$ is convex. This is because each predicate $\idstlpredpr[_{\gls*{tolpr}}^{\gls*{cscalar} - \gls*{cvec}^\intercal}]$ can be translated to a convex constraint:}
\begin{equation*}
\begin{aligned}
\idstlpredpr[_{\gls*{tolpr}}^{\gls*{cscalar} - \gls*{cvec}^\intercal}] =
P(\gls*{cscalar} - \gls*{cvec}^\intercal \gls*{statec}_k \geq 0) > 1 - \gls*{tolpr} \\
= P(\gls*{stdnormal} \leq -\frac{\gls*{cvec}^\intercal \gls*{meanvec}_k - \gls*{cscalar}}{\gls*{cvec} \gls*{covmat} \gls*{cvec}^\intercal}) > 1 - \gls*{tolpr}
= \Phi(\frac{\gls*{cvec}^\intercal \gls*{meanvec}_k - \gls*{cscalar}}{\gls*{cvec} \gls*{covmat} \gls*{cvec}^\intercal}) < \gls*{tolpr} \\
= \gls*{cvec}^\intercal \gls*{meanvec} -\Phi^{-1}(\gls*{tolpr}) \| \sqrt{\gls*{covmat}} \gls*{cvec} \|_2 < \gls*{cscalar},
\end{aligned}
\end{equation*}
\hl{where $\gls*{stdnormal} \sim \gls*{normal}(0,\gls*{eye}_{n})$ is a standard Gaussian random variable and $\Phi(z)$ is its cumulative distribution function (CDF). } The
volume of each $\gls*{poly}_i$ can be bounded with a convex polytope,
which we denote $E[\gls*{poly}_i]$.
\begin{definition}\label{def:abs} 
Given a \gls*{prtl} formula \gls*{stlformula}, $M_{\gls*{stlformula}} = (\mathbb{P},{\gls*{poly}_0},\rightarrow,\mathcal{L})$ is a Kripke structure where $\mathbb{P} = \{ \gls*{poly}_1,\dots,\gls*{poly}_{|\mathbb{P}|} \}$ are a finite set of convex partitions of the belief space, $\gls*{poly}_0 = \{ \gls*{abs}(\gls*{beliefini}) \}$, $(\gls*{poly},\gls*{poly}^\prime) \in \rightarrow$ if and only if $(E[\gls*{poly}] \cup E[\mathbb{P}_i(\gls*{poly})]) \cap (E[\gls*{poly}^\prime] \cup E[\mathbb{P}_i(\gls*{poly}^\prime)]) \neq \emptyset$ with $\mathbb{P}_i(\gls*{poly}) = \{ \gls*{poly}_i \in \mathbb{P} | \mathcal{L}(\gls*{poly}_i) = \mathcal{L}(\gls*{poly})$ and $E[\gls*{poly}_i] \cap E[\gls*{poly}] \neq \emptyset \}$, $AP = \{ \Phi_1, \dots, \Phi_N \}$ are state subformulas of \gls*{stlformula}, and $\Phi_i \in \mathcal{L}(\gls*{poly})$ if and only if $\gls*{belief} \gls*{sat} \Phi_i$ for all $\gls*{belief} \in \gls*{poly}$.   
\end{definition}

Note that any \gls*{prtl} formula \gls*{stlformula} has an unique corresponding formula $\tilde{\gls*{stlformula}}$ over the subformulas $sub(\gls*{stlformula})$. For example, for a formula $\gls*{stlformula} = (\idstlpredpr[_{\gls*{tolpr}}^{\gls*{stlfunc}_1}] \gls*{and} \idstlpredpr[_{\gls*{tolpr}}^{\gls*{stlfunc}_2}]) \gls*{until} (\idstlpredpr[_{\gls*{tolpr}}^{\gls*{stlfunc}_3}] \gls*{or} \idstlpredpr[_{\gls*{tolpr}}^{\gls*{stlfunc}_4}])$, $AP = \{ \gls*{pred}_1, \gls*{pred}_2  \}$ is the set of atomic propositions of $M_{\gls*{stlformula}}$, where $\gls*{pred}_1 = \idstlpredpr[_{\gls*{tolpr}}^{\gls*{stlfunc}_1}] \gls*{and} \idstlpredpr[_{\gls*{tolpr}}^{\gls*{stlfunc}_2}]$ and $\gls*{pred}_2 = \idstlpredpr[_{\gls*{tolpr}}^{\gls*{stlfunc}_3}] \gls*{or} \idstlpredpr[_{\gls*{tolpr}}^{\gls*{stlfunc}_4}]$. Consequently, the simplified \gls*{prtl} formula is $\tilde{\gls*{stlformula}} = \gls*{pred}_1 \gls*{until} \gls*{pred}_2$.

\begin{theorem}\label{theo:abstraction}
The abstraction of every path \gls*{runpr} in System (\ref{eq:prsystem}) that satisfies a \gls*{prtl} formula \gls*{stlformula}, i.e., $\gls*{abs}(\gls*{belief}_0)\gls*{abs}(\gls*{belief}_1)\gls*{abs}(\gls*{belief}_2)\dots$, is a path in $M_{\gls*{stlformula}}$ that satisfies the correponding $\tilde{\gls*{stlformula}}$ over $AP$.
\end{theorem}
\begin{proof}
We will prove the theorem by induction. First, note that $\{\idstlpredpr \mid \gls*{runpr} \gls*{sat}_{k} \idstlpredpr\} \subseteq \gls*{abs}(\gls*{belief}_k)$. If a trace in $M_{\gls*{stlformula}}$ satisfies $\tilde{\gls*{stlformula}}$, subsets $\tilde{\gls*{poly}}_k \subseteq \gls*{poly}_k$ also satisfy $\tilde{\gls*{stlformula}}$ since $\mathcal{L}(\tilde{\gls*{poly}}_k) = \mathcal{L}(\gls*{poly}_k)$. Moreover, by Definition \ref{def:abs}, $\gls*{abs}(\gls*{beliefini}) = \gls*{poly}_0$. Finally, if there exists $\gls*{inputc}_k \in \gls*{inputcdom}$ such that $\gls*{belief}_{k+1} = f(\gls*{belief}_k,\gls*{inputc}_k,\gls*{outputc}_k)$, then $(\gls*{abs}(\gls*{belief}_k),\gls*{abs}(\gls*{belief}_{k+1})) \in \rightarrow$. We will prove this claim by contradiction. Assume that there exists $\gls*{inputc}_k \in \gls*{inputcdom}$ such that $\gls*{belief}_{k+1} = f(\gls*{belief}_k,\gls*{inputc}_k,\gls*{outputc}_k)$, but $(\gls*{abs}(\gls*{belief}_k),\gls*{abs}(\gls*{belief}_{k+1})) \not\in \rightarrow$. Consider the case of two subformulas $\gls*{stlformula}_1$ and $\gls*{stlformula}_2$ combined with  a temporal operator in \gls*{stlformula}. In this case, there exists an instant $t_k \leq t^\prime \leq t_{k+1}$ such that $\gls*{runpr} \gls*{sat}_{(t^{\prime-})} \gls*{stlformula}_1$, $\gls*{runpr} \gls*{unsat}_{(t^{\prime-)}} \gls*{stlformula}_2$, $\gls*{runpr} \gls*{unsat}_{(t^{\prime+)}} \gls*{stlformula}_1$, and $\gls*{runpr} \gls*{sat}_{(t^{\prime+})} \gls*{stlformula}_2$. But this violates the \gls*{prtl} semantics. Therefore, this path cannot satisfy a \gls*{prtl} formula.
\end{proof} 

   
\hl{Given an abstraction $M_{\gls*{stlformula}}$ and an abstracted formula $\tilde{\gls*{stlformula}}$, we can use \gls*{bmc} tools such as NuSMV \cite{CAV02} to find a path $\gls*{polyseq}_K = \gls*{poly}_0\dots \gls*{poly}_{L-1} (\gls*{poly}_L\dots\gls*{poly}_K)^\omega$ that satisfies $\tilde{\gls*{stlformula}}$ in ($K,L$)-loop form \cite{biere1999symbolic}. Such tools find a satisfying discrete path if one exists, and otherwise indicate that no such path exists.}

\subsection{Feasibility Search}\label{sec:feas}

\hl{Existential Model Checking generates a discrete path $\gls*{polyseq}_K$ which satisfies $\tilde{\gls*{stlformula}}$.
Given such a path, feasibility search looks for a dynamically feasible path $\gls*{runpr}$ corresponding to $\gls*{polyseq}_K$. This problem can be encoded in the following feasibility problem:}
\begin{equation}\label{eq:feasproblem}
\begin{aligned}
\textbf{find} & \hspace{0.1cm} \gls*{runpr} \\
\textbf{s.t. } & \gls*{meanvec}_{0} = \gls*{meanvecini}, \gls*{covmat}_0 = \gls*{covmatini}, \\
& \text{if } L\leq K \text{ then } \gls*{meanvec}_H = \gls*{meanvec}_{\mathbb{I}(L)}, \gls*{covmat}_H \leq \gls*{covmat}_{\mathbb{I}(L)}, \\
\forall k = 1..H :~ &  \gls*{meanvec}_{k+1} = A\gls*{meanvec}_k + B\gls*{inputc}_k, \\
&    \gls*{covmat}_{k+1} = \gls*{inovar}_k - \boldsymbol{K}_k C \gls*{inovar}_k, \\
&   \gls*{cvec}_{\mathbb{I}(k)}^\intercal \gls*{meanvec} -\Phi^{-1}(\gls*{tolpr}_{\mathbb{I}(k)}) \| \sqrt{\gls*{covmat}} \gls*{cvec} \|_2 < \gls*{cscalar}_{\mathbb{I}(k)},
\end{aligned}
\end{equation}
\hl{where $\mathbb{I}$ is an increasing index function that maps the instant $k$ to the relevant constraints $\gls*{poly}_i \in \gls*{polyseq}_K$.}
\hl{In formulating this problem, we use the MLO belief dynamics (\ref{eq:beliefsyssimple}) since we do not have access to future observations, as discussed in Section \ref{sec:belief}. This simplifying assumption renders the feasibility problem computationally tractable, but limits the completeness of our approach (see Section \ref{sec:correctness}). 

Even with the MLO belief dynamics, the feasibility problem is non-convex. For this reason, we turn to rapidly exploring random tree (RRT) search. We propose modified sampling (\textsc{Sample}) and steering (\textsc{Propagate}) strategies which push the search towards trajectories within the discrete path $\gls*{polyseq}_K$. 

Furthermore, we assume that the noise covariance $\gls*{noise-std}_k$ may be state dependent. In this case, a sparse search is beneficial as it quickly samples from many nearby states. Thus we draw on techniques from SPARSE-RRT \cite{littlefield2013efficient}. This includes the function \textsc{BestNearest}, which searches for active vertices $\delta$-near the sampled state, and the function \textsc{Drain}, which removes new vertices that are too close to other active vertices.}


In \textsc{Sample}, the probability of sampling a belief state in the convex space $\gls*{poly}_k$ is inversely proportional to the number of belief states in $\mathbb{V}_{active}$ that are in $\gls*{poly}_k$. This encourages sampling points in those $\gls*{poly}_k$ that are relatively unexplored. 

\begin{algorithm}  
	\caption{\textsc{fSearch}($\gls*{polyseq}_K$,$N$)}\label{algo:feas}
    \begin{algorithmic}[1]    
    {\STATE $\mathbb{V}_{act} \gets \{ \gls*{beliefini} \}; \mathbb{V}_{inact} \gets \emptyset; \mathbb{E} \gets \emptyset; i \gets 0;$
    \STATE $\mathbb{V} = \mathbb{V}_{act} \cup \mathbb{V}_{inact}; G = \{ \mathbb{V}, \mathbb{E} \};$
			\FOR{$i = 1..K N$} 	  				
				\STATE $\gls*{belief}_{rand} \leftarrow \text{Sample}(\mathbb{V}_{act},\gls*{polyseq}_K);$\label{algo:feas_sampling1}
				\STATE $\langle \gls*{belief}_{near}, k \rangle \leftarrow \text{BestNearest}(\mathbb{V}_{act}, \gls*{belief}_{rand},\gls*{polyseq}_K,\gls*{tolfeas}_{near});$
				\STATE $\langle \gls*{belief}_{new}, \gls*{inputc} \rangle \leftarrow \text{Propagate}(\gls*{belief}_{near},\gls*{poly}_k,\gls*{poly}_{k+1});$
				\STATE $\mathbb{V}_{act} \leftarrow \mathbb{V}_{act} \cup \{ \gls*{belief}_{new} \};$  
				\STATE $\mathbb{E} \gets \mathbb{E} \cup \{ (\gls*{belief}_{near},\gls*{inputc},\gls*{belief}_{new}) \};$
				\STATE $\text{Drain}(\gls*{tolfeas}_{drain}, \gls*{belief}_{new}, G)$
				\IF{$\gls*{belief}_{new} \in \gls*{poly}_{k_{last}+1}$}         \STATE $k_{last} \leftarrow k_{last}+1;$
				\ENDIF
			\ENDFOR\label{algo:feas_sampling2}
	\STATE $\langle \gls*{runpr}, \gls*{inputc} \rangle \leftarrow \text{FeasRun}(G,V_{act});$ \label{algo:feas_cplan} 
	\RETURN $\langle \gls*{runpr}, \gls*{inputc}\rangle;$}
\end{algorithmic} 
\end{algorithm}  

The basic idea used in \textsc{Propagate}, presented as Algorithm \ref{algo:prop1}, is to account for the convexity of the constraints. Since the belief dynamics are nonlinear and underactuated, we (1) sample a belief mean from $\gls*{poly}_k \cup \gls*{poly}_{k+1}$ (line \ref{algo:prop1_sample}), (2) generate a run  towards the sampled mean (lines \ref{algo:prop1_mean1}-\ref{algo:prop1_mean2}), and (3) generate a run that minimizes the belief covariance using the first run as the nominal run (lines \ref{algo:prop1_belief1}-\ref{algo:prop1_belief2}). The first two steps essentially ensure that a trajectory is found that transitions from $\gls*{poly}_k$ to $\gls*{poly}_{k+1}$. \hl{ The last step encodes active perception: we find nominal trajectories that minimize uncertainty in the belief. 

In the case of a loop (line \ref{algo:prop1_loop}) the basic idea is to close the loop by solving the following optimization problem}
\begin{align*}
    [\gls*{inputc}_{0}^*,\dots,\gls*{inputc}_{\gls*{statecnb}-1}^*] &  = \argmin_{[\gls*{inputc}_{0},\dots,\gls*{inputc}_{\gls*{statecnb}-1}]} [\gls*{inputc}_{0}^\intercal\gls*{inputc}_{0},\dots,\gls*{inputc}_{\gls*{statecnb}-1}^\intercal\gls*{inputc}_{\gls*{statecnb}-1}] \\
    \text{s.t. } & \gls*{meanvec}_{final} - A^{\gls*{statecnb}}\gls*{meanvec}_{near} = \mathcal{C}[\gls*{inputc}_{0}^\intercal,\dots,\gls*{inputc}_{\gls*{statecnb}-1}^\intercal]^\intercal
\end{align*}
\hl{
where $\mathcal{C}$ is the controllability matrix and a closed-form solution is given on line \ref{algo:prop2_uc} of Algorithm \ref{algo:prop1}.}

\begin{algorithm}
	\caption{\textsc{Propagate}($\gls*{belief}_{near},\gls*{poly}_k,\gls*{poly}_{k+1}$)}\label{algo:prop1} 
    \begin{algorithmic}[1]   
    \IF{in loop}\label{algo:prop1_loop}
		\STATE $\gls*{belief}_{final} \gets \text{Sample}(\mathbb{V}_{act},\gls*{poly}_{k+1});$
		\STATE {$[\gls*{inputc}_{0}^\intercal,\dots,\gls*{inputc}_{\gls*{statecnb}-1}^\intercal]^\intercal \gets \mathcal{C}^{-1} (\gls*{meanvec}_{final} - A^{\gls*{statecnb}}\gls*{meanvec}_{near});$} \label{algo:prop2_uc}
		\FOR{$k = 0..\gls*{statecnb}-1$} 
			\STATE {$\gls*{inputc}_{new,k} \gets \argmin\limits_{H_u^\prime \gls*{inputc} \leq \boldsymbol{c}_u(\gls*{belief}_{new,k})} \|  \gls*{inputc} + \gls*{inputc}_{k} \|_1;$}\label{algo:prop2_ucons}
			 \STATE {$\gls*{belief}_{new,k+1} \gets \tilde{f}(\gls*{belief}_{new,k},\gls*{inputc}_{new,k});$}
		\ENDFOR 
		{\RETURN $\langle \gls*{belief}_{new,\gls*{statecnb}}, \gls*{inputc}_{new} \rangle;$}
	\ELSE{}
        \STATE {$\langle \gls*{meanvec}_{final}, T \rangle \gets $Sample$(\gls*{poly}_k,\gls*{poly}_{k+1})$
        \STATE $\gls*{belief}_{final} \gets \gls*{normal}(\gls*{meanvec}_{final},\boldsymbol{0}_{n\times n});$} \label{algo:prop1_sample} 
    	\STATE {$F \leftarrow \text{LQR}(A,B,Q,R), \gls*{belief}_{0}^\prime \gets \gls*{belief}_{near};$ }\label{algo:prop1_mean1} 
    		\FOR{$k = 0..T-1$} 
    			\STATE {\small$\gls*{inputc}_{k}^\prime \gets \argmin\limits_{H_u^\prime \gls*{inputc} \leq \boldsymbol{c}_u(\gls*{belief}_{k}^\prime)} \|  \gls*{inputc} + F(\gls*{meanvec}_{k}^\prime - \gls*{meanvec}_{final}) \|_1;$}
    			\STATE {\small$\gls*{belief}_{k+1}^\prime \gets \tilde{f}(\gls*{belief}_{k}^\prime,\gls*{inputc}_{k}^\prime);$ }\label{algo:prop1_mean2}
    		\ENDFOR 
    		\STATE {$F_k \gets \text{B-LQR}(T,\tilde{f},Q,Q_f,R), \gls*{belief}_{new,0}^\prime \gets \gls*{belief}_{near};$ }\label{algo:prop1_belief1}
    		\FOR{$k = 0..T-1$} 
    			\STATE {$\gls*{inputc}_{new,k} \gets \argmin\limits_{H_u^\prime \gls*{inputc} \leq \boldsymbol{c}_u(\gls*{belief}_{new,k})} \|  \gls*{inputc} + F_k(\gls*{belief}_{k}^\prime - \gls*{belief}_{final}) \|_1;$}
    			\STATE {$\gls*{belief}_{new,k+1} \gets \tilde{f}(\gls*{belief}_{new,k}^\prime,\gls*{inputc}_{k}^\prime);$} \label{algo:prop1_belief2}
    		\ENDFOR 
    		{\RETURN $\langle \gls*{belief}_{new,T}, \gls*{inputc}_{new} \rangle$}
		\ENDIF
	\end{algorithmic}
\end{algorithm} 


Algorithm \ref{algo:feas} has similar properties to SPARSE-RRT, including asymptotic near-optimality:
\begin{theorem}\label{theo:fsearch}
Given a discrete path $\gls*{polyseq}_K$ in $K$-loop form, Algorithm \ref{algo:feas} finds a path $\gls*{runpr}_H$ of System (\ref{eq:prsystem}) only if \hl{Eq. \ref{eq:feasproblem}} has a solution. Moreover, if there exists an optimal path $\gls*{runpr}^* \in \gls*{polyseq}_K$ of System (\ref{eq:prsystem}), Algorithm \ref{algo:feas} will eventually generate \hl{a solution to Eq. \ref{eq:feasproblem}}.   
\end{theorem}
\begin{proof}
SPARSE-RRT \cite{littlefield2013efficient} assumes that (1) the state space is sampled uniformly and (2) \textsc{Propagate} randomly selects controls of propagation. It is easily seen that Assumption (1) is asymptotically guaranteed. Assumption (2) is guaranteed by randomly selecting target states $\gls*{belief}_{final}$ and duration $T$. Thus our modified propagation algorithm gives the dispersion necessary to achieve asymptotic optimality. Therefore, from \cite[Lemma 2]{littlefield2013efficient}, Algorithm \ref{algo:feas} will eventually generate a path $\gls*{runpr}$ in $\gls*{polyseq}_K$, given that there exists an optimal path $\gls*{runpr}^*$  in $\gls*{polyseq}_K$. The proof of soundness is trivial since \textsc{Propagate} only generates valid path segments.    
\end{proof}

\hl{Feasibility search returns a trajectory (a sequence of states and control inputs) that satisfies the given PRTL specification. Since the feasibility problem is solved with sampling-based search, this trajectory is not necessarily unique; however, it is guaranteed to satisfy the specification.}

\subsection{Iterative Deepening Search}

If feasibility search does not find a satisfying run, we consider the given abstract path $\gls*{polyseq}_K$ to be infeasible. We then use $\gls*{polyseq}_K$ as a counterexample, and generate a new discrete path. Specifically, we generate a new Kripke structure by taking the product of $M_{\gls*{stlformula}}$ and the complement of $\gls*{polyseq}_K$. In this way, new discrete paths found by \gls*{bmc} over the new Kripke structure will avoid the infeasible counterexample. 

\begin{algorithm}
	\caption{\textsc{idPRTL}}
	\label{algo:idprtl}
    \begin{algorithmic}[1]
		\STATE
		$\langle M_{\gls*{stlformula}}, \tilde{\gls*{stlformula}} \rangle \leftarrow \text{abstract}(\gls*{stlformula}, sys);$\label{alg:idstl_init}
        \WHILE{$\text{\gls*{bmc}}(M_{\gls*{stlformula}}, \tilde{\gls*{stlformula}}) \neq \emptyset$}
            \STATE $\gls*{polyseq}_k \leftarrow \text{\gls*{bmc}}(M_{\gls*{stlformula}}, \tilde{\gls*{stlformula}}) $  
            \STATE $\langle \gls*{runpr}, \gls*{inputc} \rangle \leftarrow \text{fSearch}(\gls*{polyseq}_K);$
            \IF{feasible} 
                \RETURN $\langle \gls*{runpr}, \gls*{inputc} \rangle$;
            \ENDIF
            \STATE $M_{\gls*{stlformula}} \leftarrow M_{\gls*{stlformula}} \times \text{toKripke}(\gls*{polyseq}_K)^c;$  \label{algo:idprstl_cex}
		\ENDWHILE
		\RETURN infeasible
  \end{algorithmic}
\end{algorithm}

This process, outlined in Algorithm \ref{algo:idprtl}, continues until either a satisfying trajectory of System (\ref{eq:prsystem}) is found or \gls*{bmc} indicates that there is no satisfying discrete path. We call this Iterative Deepening PRTL because it is inspired by iterative deepening graph search \cite{kira:Russell:2009}. Each iteration of \gls*{bmc} is like a depth-limited depth-first search over the space of possible discrete plans. Repeating this procedure with the new product Kripke structure is analogous to increasing the depth of the search. 

\hl{
\section{Soundness and Completeness}\label{sec:correctness}

First, we prove the soundness of our approach as follows:
\begin{theorem}\label{theo:soundness}
Given a stochastic system (\ref{eq:prsystem}), a \gls*{prtl} formula $\gls*{stlformula}$, and a prior belief state $\gls*{beliefini} \sim \gls*{normal}(\gls*{meanvecini},\gls*{covmatini})$, Algorithm \ref{algo:idprtl} finds a path $\gls*{runpr}_H$ only if $\gls*{runpr}_H \gls*{sat} \gls*{stlformula}$. 
\end{theorem} 
\begin{proof}
Assume that Algorithm \ref{algo:idprtl} finds a path $\gls*{runpr}_H$. From Theorem \ref{theo:fsearch}, this run is in a discrete path that satisfies the specification. From Theorem \ref{theo:abstraction}, the trace of this run is a trace of $M_{\gls*{stlformula}}$ which satisfies the formula \gls*{stlformula}. Thus $\gls*{runpr}_H$ is a path of System (\ref{eq:prsystem}) that satisfies $\gls*{stlformula}$. 
\end{proof}

We prove probabilistic completeness as follows, by showing that the nonexistence of a solution of Algorithm \ref{algo:idprtl} is evidence that no solution exists in the following sense:

\begin{theorem}\label{theo:completeness}
Given a sufficient large sample upper bound $N$, a stochastic system (\ref{eq:prsystem}), a \gls*{prtl} formula $\gls*{stlformula}$, and a prior belief state $\gls*{beliefini} \sim \gls*{normal}(\gls*{meanvecini},\gls*{covmatini})$, if Algorithm \ref{algo:idprtl} does not find a path $\gls*{runpr}_H$, then the probability that any feasible $\gls*{runpr}_H \gls*{sat} \gls*{stlformula}$ exists is less than $0.5$. 
\end{theorem} 
\begin{remark}[Sample Upper Bound]
    Since feasibility search is a stochastic algorithm, we cannot guarantee that there is a particular minimum value of $N$; however, as $N \rightarrow \infty$, Theorem \ref{theo:completeness} holds. 
\end{remark}

We need the following lemma to prove the theorem:}
\begin{lemma}\label{lem:mlo}
Given a discrete path $\gls*{polyseq}_K$ in $K$-loop form, if there does not exist a path $\gls*{runpr} \gls*{sat} \gls*{stlformula}$ which satisfies the \gls*{mlo} dynamics (\ref{eq:beliefsyssimple}), the probability that there exists any such path under the true belief dynamics (\ref{eq:beliefsys}) is less than $0.5$. 
\end{lemma}   
\begin{proof}
Assume that there is a satisfying path in $\gls*{polyseq}_K$ under the true belief dynamics (\ref{eq:beliefsys}), but no such path under the \gls*{mlo} assumption. For this to be so, the initial system state $\gls*{statec}_0$ must be backwards reachable from $E[\gls*{poly}_k]$ for any $k > 0$, but the initial belief mean $\gls*{meanvecini}$ must not be. For discrete-time linear control systems, the backwards reachability of convex polytopic spaces is a sequence of convex polytopes since  linear transformations of convex polytopes are convex polytopes \cite{ziegler2012lectures}. Hence, the mean $\gls*{meanvecini}$ of the initial belief state must be outside of all those convex polytopes, meaning that the probability of the state reaching $E[\gls*{poly}_k]$ is equal to the probability of being in such a convex polytopes. Since these polytopes are convex and do not include the mean, they must have probability less than $0.5$.    
\end{proof}  
\hl{
Now, we can continue the proof.
\begin{proof}
From Theorem \ref{theo:abstraction}, if there is no trace in the abstraction $M_{\gls*{stlformula}}$ that satisfies the formula \gls*{stlformula}, there is no trace in System (\ref{eq:prsystem}) either. Moreover, from Theorem \ref{theo:fsearch}, if there exists a trace in the abstraction but the feasibility search (Algorithm \ref{algo:feas}) does not find a path, assuming that we have a sufficient large sample upper bound $N$, there is no path that satisfies the specification with probability higher than $0.5$ (Lemma \ref{lem:mlo}).
\end{proof}

\begin{remark}[Completeness]
From Lemma \ref{lem:mlo} above, we see that the 0.5 probability restriction on our completeness guarantee stems from the use of MLO belief dynamics (\ref{eq:beliefsyssimple}). If we were able to plan with the full belief dynamics (\ref{eq:beliefsys}) this probability could be raised to 1. Unfortunately, using the full belief dynamics is not feasible, as future observations are not known a priori. 

One approach to improve this 0.5 probability guarantee would be to define simplified belief dynamics based on sampling future observations. Such techniques have been proposed for belief-space planning in \cite{hadfield2015modular,van2012efficient,indelman2014planning,indelman2016towards}, and remain an area of active research.  
\end{remark}

\begin{remark}[Complexity]
The complexity of \textsc{idPRTL} depends on the \gls*{prtl} formula complexity and the parameter $N$. First, the worst case number of symbols in the abstracted system (Def. \ref{def:abs}) is exponential in number of predicates, $O(2^{|\gls*{predset}|})$. State-of-the-art \gls*{bmc}  solvers are linear in the number of symbols and the length of bound $K$, $O(K)$. Finally, the complexity of each feasibility search query is $O(N \log N)$ \cite{li2015sparse}.
\end{remark}
}


\section{Example}\label{sec:mot_ex}

In this section, we apply our planning framework to automated infrastructure inspection with a quadrotor. The inspection task involves collecting images of key points. While the locations of these key points are often known a priori, imperfect localization requires that the quadrotor must maintain a belief over its location. Furthermore, when the quadrotor is directly over the power line it can use a camera for precise localization, while it must otherwise rely on less accurate GPS.  

A simple inspection task might be articulated as follows:
``Staying away from the power line, take pictures of the top of two poles and return to the charging station.'' This task, illustrated in Figure \ref{fig:quad_demo}, can be written in \gls*{prtl} as follows::
\begin{align}\label{eq:example_spec}
\gls*{stlformula} = & \gls*{always} \gls*{stlformula}_{safe} \gls*{and} \gls*{eventually} \Big((\gls*{stlformula}_{pole_1} \gls*{or} \gls*{stlformula}_{Pline}) \gls*{until} \gls*{stlformula}_{pole_2}\Big) \gls*{and} \gls*{eventually}\gls*{always} \gls*{stlformula}_h,  
\end{align}
where $\gls*{stlformula}_{safe}$ denotes avoiding collisions with the powerline (grey box in Figure \ref{fig:quad_demo}), $\gls*{stlformula}_{Pline}$ indicates flying over the powerline, $\gls*{stlformula}_{pole_*}$ indicate visiting the key points, and $\gls*{stlformula}_h$ indicates returning ``home'' to the charging station. \hl{All predicates are defined with $\epsilon = 0.05$.}

For the quadrotor dynamics, we follow \cite{zhou2014vector} in modeling the quadrotor with a chain of integretors using differential flatness. We model the observations $\gls*{outputc}$ as follows, where $\gls*{statec} = [x, y, z]^\intercal$ is the position of the quadrotor:
\begin{align*}
    & \gls*{outputc} = \gls*{statec} + w(\gls*{statec}) \\
    & w(\gls*{statec}) \sim \mathcal{N}(0, \min(\gls*{std}_{gps}^2, \gls*{std}_{camera}^2(\gls*{statec})) \\
    & \gls*{std}_{camera}^2(\gls*{statec}) = (y-y_{center} )^4 + (z-z_{top})^4 + \gls*{std}_{min}^2
\end{align*}

Essentially, the quadrotor has low localization uncertainty while it is directly over the power line ($\gls*{std}_{camera}^2$), but high localization uncertainty otherwise ($\gls*{std}_{gps}^2$). \hl{After specifying an initial belief, \textsc{idPRTL} returns the trajectory shown in blue, where the shaded region indicates the covariance. We know from Theorem \ref{theo:soundness} that this trajectory satisfies the specification.}

Figure \ref{fig:quad_demo} illustrates three important features of our approach. First, the resulting trajectory satisfies the specification, as both points of interest (green) are visited before returning to the charging station (red). Second, the specification (\ref{eq:example_spec}) is non-convex, as it contains a disjuction and nested temporal operators. This means that existing sampling-based approaches such as \cite{dey2016fast} cannot be used in this scenario. Finally, the resulting trajectory illustrates the use of active perception. The quadrotor first flies a conservative distance above the powerline, approaching it from above so that the camera can be used to reduce uncertainty before moving close to the key points.

\begin{figure}
	\centering
	\includegraphics[width=0.7\linewidth]{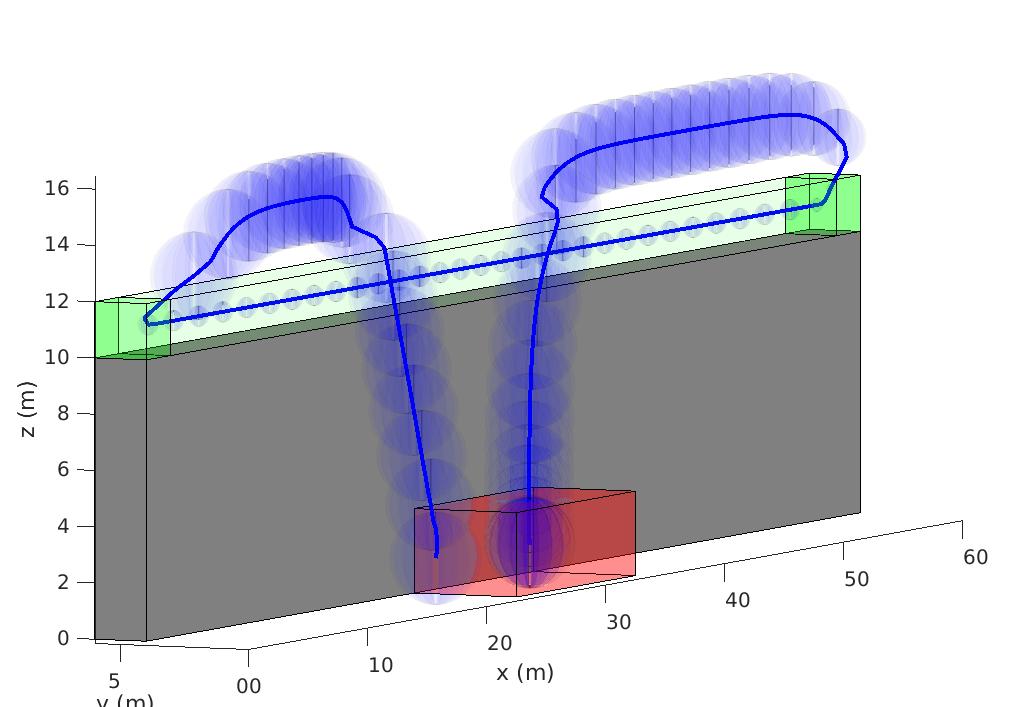}
    \caption{A quadrotor performing an inspection task must visit two key points (green) and avoid an obstacle (gray) before returning to a charging station (red). The belief trajectory (blue) indicates the use of active perception: the quadrotor flies high above the powerline to improve confidence in its belief before approaching the key points. }
    \label{fig:quad_demo}  
\end{figure}

\section{Conclusion and Future Work}\label{sec:conclusion}

We presented a framework for controller synthesis from \gls*{prtl} specifications in the belief space that combines the advantages of Bounded Model Checking and sampling-based motion planning. Our framework allows for active perception in complex tasks involving nonconvex \gls*{prtl} specifications. We demonstrated the efficacy of our approach on a simulation of a quadrotor power line inspection task. 


%







\bibliographystyle{IEEEtran}
\bibliography{library}
%



\end{document}